\documentclass[11pt, draftclsnofoot, onecolumn]{IEEEtran}  % draftclsnofoot

\usepackage{amsmath,graphicx}
\usepackage{amssymb}
\usepackage{algpseudocode}
\usepackage{cite}
\usepackage{amsfonts}
\usepackage{graphicx}
\usepackage{fancyhdr}  %
\usepackage{cases}
\usepackage{extarrows}
\usepackage{algorithm}
\usepackage{multirow,tabularx}
\usepackage{mathtools}
\usepackage{xcolor}
\usepackage[english]{babel}

\usepackage{caption}

\usepackage{float}
\usepackage{cases}
\usepackage{subfigure}

\usepackage{footnote}
\makesavenoteenv{tabular}
\makesavenoteenv{table}

\IEEEoverridecommandlockouts   % No this comments, it will do not show thanks content;
\usepackage{amsmath,amsthm}

\usepackage{hyperref}% http://ctan.org/pkg/hyperref
\hypersetup{%
  bookmarks=true
}

\begin{document}

\title{
A Joint Energy and Latency Framework for Transfer Learning over 5G Industrial Edge Networks}

\author{

\IEEEauthorblockN{Bo~Yang,~\IEEEmembership{Member,~IEEE},
 Omobayode Fagbohungbe, Xuelin Cao, Chau Yuen,~\IEEEmembership{Fellow,~IEEE}, 
  Lijun Qian,~\IEEEmembership{Senior Member,~IEEE}, Dusit Niyato,~\IEEEmembership{Fellow,~IEEE}, and Yan Zhang,~\IEEEmembership{Fellow,~IEEE}}

 \thanks{B. Yang, X. Cao, and C. Yuen are with the Engineering Product Development Pillar, Singapore University of Technology and Design, Singapore 487372 (e-mail: bo$\_$yang, xuelin$\_$cao, yuenchau@sutd.edu.sg).}% 
\thanks{Omobayode Fagbohungbe and Lijun Qian are with the Department of Electrical and Computer Engineering and CREDIT Center, Prairie View A$\&$M University, Texas A$\&$M University System, Prairie View, TX 77446, USA (email: ofagbohungbe@student.pvamu.edu, liqian@pvamu.edu). }
\thanks{D. Niyato is with the School of Computer Science and Engineering, Nanyang Technological University, Singapore (e-mail: dniyato@ntu.edu.sg).}
\thanks{ Y. Zhang is with University of Oslo, Norway; and also
with Simula Metropolitan Center for Digital Engineering, Norway (e-mail:
yanzhang@ieee.org).}
}

\maketitle

\begin{abstract}
In this paper, we propose a transfer learning (TL)-enabled edge-CNN framework for 5G industrial edge networks with privacy-preserving characteristic. {In particular, the edge server can use the existing image dataset to train the CNN in advance, which is further fine-tuned based on the limited datasets uploaded from the devices. With the aid of TL, the devices that are not participating in the training only need to fine-tune the trained edge-CNN model without training from scratch.}
Due to the energy budget of the devices and the limited communication bandwidth, a joint energy and latency problem is formulated, which is solved by decomposing the original problem into an uploading decision subproblem  and a wireless bandwidth allocation subproblem. {Experiments using ImageNet demonstrate that the proposed TL-enabled edge-CNN framework can achieve almost $85\%$ prediction accuracy of the baseline by uploading only about $1\%$ model parameters, for a compression ratio of $32$ of the autoencoder.}

\begin{IEEEkeywords}
Industrial IoT, transfer learning, convolutional neural network, 5G, multi-access edge computing
\end{IEEEkeywords}
\end{abstract}

\section{Introduction}
\IEEEPARstart{I}{nternet}-of-Things (IoT) has become encouraging by connecting enormous devices, such as sensors, CCTV, smart glasses, and autonomous vehicles. As a subset of IoT and an emerging domain, the Industrial IoT (IIoT) inter-connects the humongous amount of smart industrial units and equipment capable of environmental monitoring, data collection, and processing~\cite{IIoT_editor2}. As reported by IBM, the amount of connected-equipment will balloon to 50 billion, generating over 40 trillion Gbits data by 2020~\cite{5G}. Different from the conventional cellular networks (e.g., 3G and 4G), the next-generation wireless communication, also known as (a.k.a.) the fifth-generation (5G), is expected to provide a connection for a large volume of IIoT devices with ultra-low latency yet high reliability ubiquitously. 

During recent years, the exponential growth of promising technologies such as artificial intelligence (AI) will usher in unprecedented paradigm shifts from connected devices to connected intelligence~\cite{5G_AI,yc_its}. In this context, it is envisioned that future IIoT services will emerge intelligently through hyper-connectivity involving 5G and AI, which will pave the way for enabling efficient automation production by interconnecting isolated industrial assets. As a crucial technique achieving image classification~\cite{AlexNet}, %{image_classification01}},
 convolutional neural network (CNN) has significantly promoted developments of visual processing tasks such as object detection~\cite{{object_detection01},{object_detection02},{object_detection03}}. By feeding the collected images into the well trained CNN model, the spatial and temporal dependencies of data can be captured through the application of relevant filters.
 
Since relatively high computational complexity is introduced during the CNN model's training based on the big dataset, multi-access edge computing (MEC) becomes promising in IIoT by reducing the devices' computation burdens by ``moving" the computation tasks to the edge server~\cite{MEC_editor1_01,yc_tii}. Specifically, with the help of MEC, the raw images data collected by the devices can be offloaded to the MEC server (MES) co-located with the base station (BS) for CNN training, which, however, will lead to the data privacy concern and introduce considerable training delay, especially when the edge server is overwhelmed by the computing load~\cite{yb_wcm}. {Such issues may become particularly critical in some applications such as smart grid IoT system, where the smart meter users with sensitive data are usually densely deployed and thus incur high computation and communication costs. To achieve individual privacy, a fog-enabled privacy-preserving data aggregation scheme (FESDA) was proposed to encrypt the metering data and send to the control center~\cite{FESDA}. }
%\hspace{-2mm}

Generally, to achieve the privacy-preserving during the model training, federated learning proposed by Google allows the end-devices to collaboratively train a shared learning model while preserving all training data on their own sides~\cite{FL_google}. By joint optimizing the computing and communication resources such as the global aggregation frequency, the federated learning performance can be improved~\cite{{FL_jsac1},{FL_jsac}}. While interesting, Zhu  \textit{et al}.~\cite{privacy} stated that this kind of implementations may still suffer from leaking gradient problem associated with federated learning. This ``leaking gradient problem" can be avoided with autoencoder, which protects data privacy without the necessity of encryption~\cite{autoencoder}. Moreover, in real-world IIoT applications, the training dataset and the dataset in the target domain may not be within the same feature space {or follow a different data distribution~\cite{TL_distribution1}. For instance, the input spaces may be different between the source and the target tasks~\cite{TL_distribution2}.} To address this problem, the transfer learning (TL, a.k.a. knowledge transfer) has emerged as a promising learning method in which the knowledge learned from the previously developed models in the {source domain} can be leveraged in some related training tasks to decrease the model convergence time in the {target domain}~\cite{TL_survey}.

\textit{Our Vision:} {In this vein, AI functioning on edge will play a crucial role in image classification and become a general trend~\cite{Chen_ioffloading}, and how to improve the CNN performance in 5G edge IIoT era with data privacy-preserving is the focus of this article.} This work aims to provide a comprehensive TL-enabled edge-CNN framework with joint consideration of model training and fine-tuning, privacy-preserving, energy consumption and latency involved:

\begin{itemize}
\item \textit{A CNN model is trained or fine-tuned at the edge server to ease the computation burdens on the devices via TL}. {In this paper, transfer learning refers to the method of fine-tune a pretrained deep learning model with new data that may have different distribution. During the CNN training, the dataset collected or generated by the devices is prone to be kept locally due to the data privacy. As a result, the labeled data may be too limited for the edge server to train a CNN from scratch. One promising way to address this issue is TL, which enables the edge server to use the existing image dataset (e.g., ImageNet) to train a CNN in advance. Then the trained CNN can be fine-tuned based on the limited datasets uploaded from the devices to achieve reasonable inference performance. Notably, the proposed framework is applicable for both cases, and this paper assumes that CNN is trained from scratch at the edge server.}

\item \textit{Each device trains an autoencoder and uploads the latent vectors to the edge server for privacy-preserving}. 
{Instead of
reduction of the data redundancy using traditional compression methods, autoencoder can extract critical features from the raw data and encode the features into the latent vector. Specifically, in addition to data compression, autoencoder can also encrypt the data by transforming the raw data into latent vectors, which further enhance the security of data since an adversary could not reconstruct the raw data from the latent vector without knowing the exact structure and weights of the pre-trained autoencoder~\cite{auto}.
Compared to the differential privacy~\cite{ICLR} and secure aggregation~\cite{secure_aggregation}, the proposed framework can protect data privacy by transmitting latent vectors without introducing additional cost of encryption or secret sharing.}
\end{itemize}

Specifically, since the distribution of data collected at different IIoT devices is usually non-independent-and-identically-distributed (non-i.i.d.), it is foreseeable that the devices with various data sources are prone to participate in the CNN training to improve the image classification performance. While impressive, the more devices involved the CNN training (i.e., the more devices need to upload their latent vectors to BS), then less wireless bandwidth can be allocated for each of the devices. This leads to a \textit{wireless transmission bottleneck} in the proposed edge-CNN framework. With the aid of TL, the devices not participating the edge-CNN training %can directly use the trained CNN model to perform the inference without fine-tuning, while the remaining devices 
should fine-tune the trained CNN model before the local inference. Since the IIoT devices are usually energy-constrained, we note that performing fine-tuning at devices should be avoided if it costs too much energy. In this context, the BS has to carefully select the devices participating in the CNN training and allocate appropriate wireless bandwidth for the involved devices to minimize the computing latency and communication latency.

\textit{Contributions and Organization:}  To achieve efficient image classification via CNN training with transfer learning 
in 5G-driven IIoT ecosystems, inherent concerns associated with energy-constrained computing and secure communication under limited bandwidth will present stumbling blocks toward the envisioned goals (e.g., real-time data analytic).
In this paper we explore the challenges raised by investigating the intertwined relationships among the involved energy consumption, completion latency, and the CNN training quality. {To the best of our knowledge, this is the \textit{first} effort to investigate and optimize the CNN training and inference via TL for energy-constrained IIoT edge networks under limited wireless bandwidth.}  The specific contributions and findings of this paper include:

 \begin{itemize}
\item {\textbf{TL-enabled Edge-CNN Framework:}} We develop an edge-CNN framework with privacy-preserving via TL over 5G-aided IIoT networks. Specifically, to achieve privacy-preserving, the device trains the autoencoders individually and compresses the image into the latent vectors that usually contain critical features learned from the raw images~\cite{qian_conference}. Then the latent vectors and the corresponding labels will be uploaded to the edge server for CNN training. On receiving the latent vectors and labels, the edge server trains the edge-CNN model and broadcasts the well-trained CNN model to all the devices. With the aid of TL, the devices that are not participating in the training only need to fine-tune the received edge-CNN model with their datasets without training from scratch. This indicates that the more devices participating in the training process, the better generalization of the edge-CNN can be achieved.

\item \textbf{Energy-Latency-and-Learning Tradeoffs:} 
The \textit{energy-latency-and-learning} tradeoffs consists of two parts, namely \textit{energy-learning} tradeoff and \textit{latency-learning} tradeoff. 
These two tradeoffs demonstrate the relationship between the device's energy consumption and completion latency achieving image classification, and the CNN learning performance. Specifically, the larger compression ratio (leading to the worse inference accuracy) will lead to less communication latency and energy consumption involved since more wireless bandwidth becomes available. Since the latency involved before the CNN training consists of two parts: computing latency and communication latency, %this tradeoffs give rise to a third tradeoff   
 they form another tradeoff, namely \textit{computing-communication} latency tradeoff, as highlighted in Fig.~\ref{model}(b). That is to say, only when all devices finish the local computing (i.e., compress the raw images into the latent vectors) and wireless communication (i.e., upload the latent vectors to the BS) simultaneously, the ``straggler effect" issue that is described in \textbf{Remark~\ref{R2}} can be avoided.

\item \textbf{Joint Energy-Latency Optimization:} 
To trace the energy-latency-and-learning tradeoffs, in our proposed approach we fix the learning metric (i.e., the compression ratio) and then solve for the energy-latency problem. Specifically, considering that the IIoT devices are usually energy-constrained and wireless bandwidth is limited, we formulate a joint energy and latency optimization problem to find an optimal uploading decision and wireless bandwidth allocation minimizing the weighted-sum cost. %since communications and computing are intertwined in 5G-driven IIoT networks with MEC setup. 
We decompose the original problem into two subproblems and derive the solution, % achieve the best {computing-communication} tradeoff, 
 which suggests allocating more bandwidth to the involved devices with worse channel conditions or weaker computation capabilities.

\item \textbf{Performance Evaluation:} 
Experiments using ImageNet dataset demonstrate the superiority of the proposed edge-CNN framework in terms of CNN performance (including the testing accuracy and training time) and normalized weighted-sum cost.

\end{itemize}

The rest of the paper is organized as follows. We start by reviewing the related works %and discussing the pros and cons 
in Section~\ref{related_works}. Section~\ref{framework} describes  the proposed framework and system model in detail. In Section~\ref{optimization}, we present the problem formulation, followed by the problem decomposition achieving the optimal latent vector uploading decision and wireless bandwidth allocation. Experimentation results are illustrated in Section~\ref{result}. Conclusions is finally drawn in Section~\ref{concusion}.

\section{Related Works}
\label{related_works}
Some previous works dealt with challenges on the model training (e.g., CNN) from various perspectives. Here, we analyze and distinguish our work from these existing works.
%\hspace{-40mm}
\subsection{{CNN Training Architecture Design}}
%\textcolor{blue}{\subsection{CNN Training Architecture Design}}
During the past decade, CNN based sophisticated learning architectures such as ImageNet~\cite{GoogleNet} have achieved significant progress in image classification. To improve the performance of the conventional CNN training architecture, some researchers aim at combining the CNN architecture with the advanced deep learning models~\cite{{advanced01},{advanced02}}. %{advanced03}}. 
For example, Wang \textit{et al}.~\cite{advanced01} combined CNN with the recurrent neural network (RNN) for multi-label image classification. Krizhevsky \textit{et al.} \cite{advanced02} %and Wang \textit{et al}.~\cite{advanced03} 
combined CNN with autoencoder to achieve face rotation and intrinsic transformations for objects. On another note, some works focused on improving CNN training by modifying the current CNN architectures for some specific tasks~\cite{vnet}. %For instance, U-Net~\cite{unet} and V-Net~\cite{vnet} were proposed to accomplish the medical image processing in which the training occurred at the cloud, which, however, generally needs hours or even several days. %\subsubsection{Transfer Learning and Its Implementation over Wireless Networks}
In real-world applications, the training dataset and the dataset in the future target domain may not be within the same feature space but follow different distributions.
 By exploring the TL advantage, the training efficiency can be significantly improved by avoiding expensive data labeling efforts. In this context, introducing TL into CNN training attracts extensive attention and achieves great success in different tasks such as image recognition and object detection~\cite{{TL02},{TL03},{TL04}}. For example, 
 Oquab \textit{et al.}~\cite{TL02} designed a method to reuse ImageNet-trained layers of CNN to compute mid-level image representation for images in a different dataset. To bridge the chasm between image classification and object detection, Girshick \textit{et al.}~\cite{TL03} proposed a CNN-based detection algorithm, called ``Regions with CNN features" (R-CNN), that was composed of ``supervised pre-training" and ``domain-specific fine-tuning".
In~\cite{TL04}, Shin \textit{et al.} verified that the TL from pre-trained ImageNet CNN models could be valuable on the medical image datasets.

\subsection{Transfer Learning over Wireless Networks}
 With the successes of MEC applied in the 5G era,  considerable attention has been paid toward improving the TL performance with the aid of wireless communication and computing technologies. For example, some existing works, {\cite{TL_wireless01,TL_wireless02, TL_wireless03, TL_wireless04, TL_wireless05, TL_wireless06, TL_wireless07},} investigated the issues related to the implementation of TL over wireless networks. Specifically, Chen \textit{et al.}~\cite{TL_wireless01} designed a task allocation scheme assigning tasks to edge devices according to their importance to maximize the overall decision performance.  Ba{\c{s}}tu{\u{g}} \textit{et al.}~\cite{TL_wireless02} proposed a transfer learning-based caching procedure by exploiting the prior information, which is then incorporated in the target domain to cache strategic contents at the small cells optimally. {Shen \textit{et al.}~\cite{TL_wireless03} presented a TL method via adapting a pre-trained machine learning model to the new task with only a few unlabeled training samples for resource allocation in wireless networks. Liu \textit{et al.}~\cite{TL_wireless04} investigated a deep transfer learning approach to
address tag signal detection problem for AmBC systems, by transferring the knowledge learned from one tag detection task under the offline channel coefficients to another different but related tag detection task
in real-time. Furthermore, Yuan \textit{et al.}~\cite{TL_wireless05} introduced TL to achieve fast adaptation with the limited new labeled data on the beamforming design when the distribution of testing wireless environments changes. To improve the robustness of spectrum sensing, Peng \textit{et al.}~\cite{TL_wireless06} proposed to incorporate TL to adapt the learned models to new communications settings. In~\cite{TL_wireless07}, a TL-based framework was proposed for rotor-bearing system fault diagnosis
under varying working conditions to solve the problem of limited available training data in the target domain.}

Notably, we found that in the current literature, the devices usually need to train a shallow model and transfer specific layers or even the whole model to the target tasks at the edge server or cloud. %~\cite{TL05}, 
 This procedure may still suffer from leaking data privacy problem, and pose a significant burden on the power-constrained IIoT devices as well as wireless communications since a large volume of knowledge is involved. As a result, these prior works are not applicable in the practical 5G-aided IIoT networks since devices' differential computation capability (such as the energy, storage space, and CPU resource) and the dynamics of the wireless channels were not considered yet. With the implementation of the TL algorithms on the MEC setup, the CNN training will encounter various energy and delay, which will, in turn, affect the CNN training performance.

\section{Proposed Edge-CNN Framework}
\label{framework}
In this section, we first introduce the proposed framework and the critical training principles, followed by the detailed energy  model and latency model.

\begin{figure}[t] \vspace{-0.0cm} \hspace{-0.0cm}
            \centering
            \captionsetup{font={footnotesize }}
            \includegraphics[width=4.5in, height=3.5in]{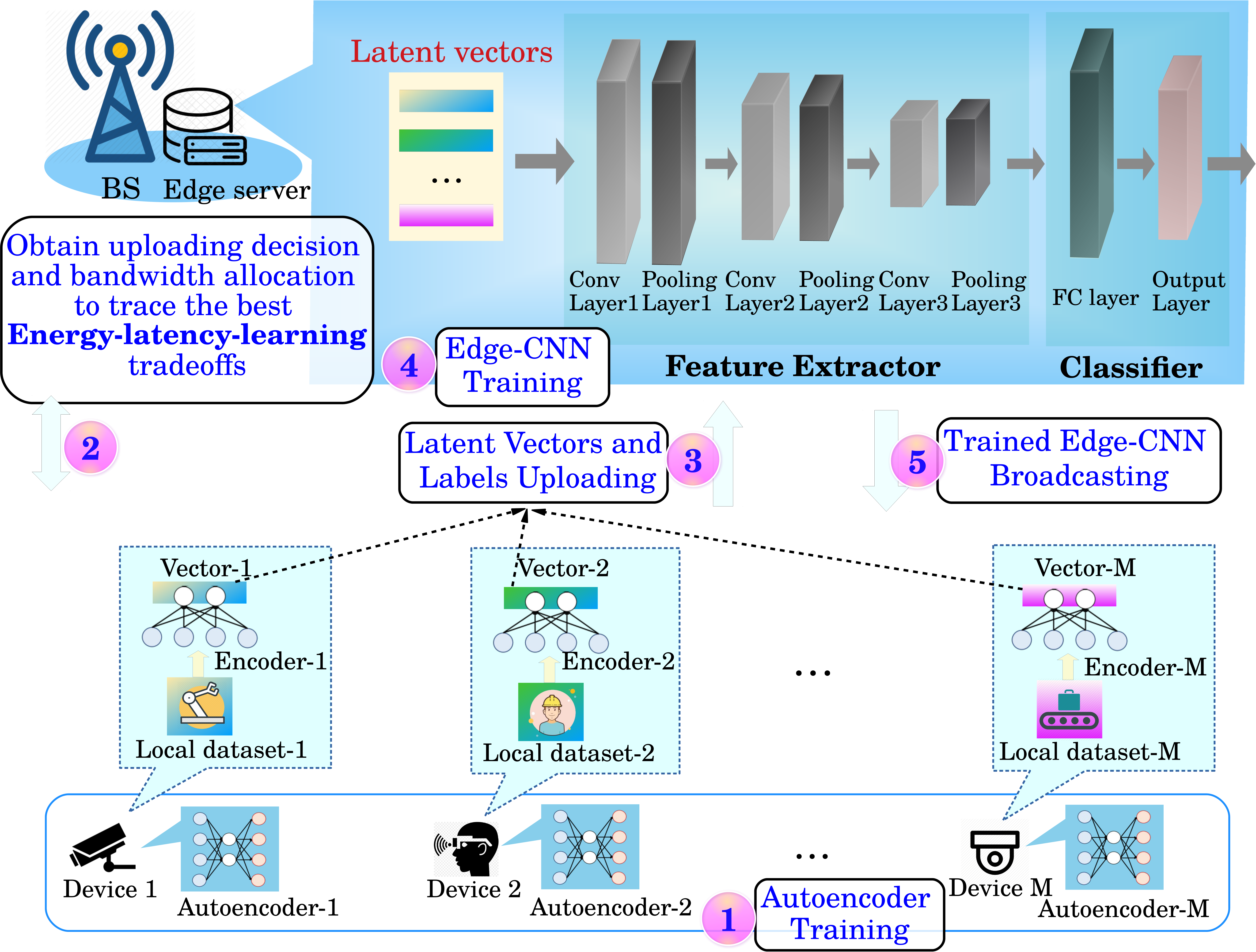} % \vspace{-0.2cm}
             \caption{An illustrative example of the proposed framework for image classification, where there exists $M$ IIoT devices. The edge-CNN classifier layers contain a mix of convolution, max-pooling and fully connected layers. } %and the edge-CNN is composed with three ``Conv+Pooling" layers.}
            \label{fig:1}
\end{figure}\vspace{-0.0cm} \vspace{-0.0cm}

\subsection{Overview of Proposed Framework}
 A general structure of the proposed framework is highlighted in Fig.~\ref{fig:1}. Suppose that there are total $ M $ IIoT devices, each device contains embedded sensors such as video cameras and embedded computing devices such as NVIDIA Jetson Nano, while the edge server has strong computational resources and large storage. The devices set is given by ${\cal M}=\{D_1, D_2, ..., D_M\}$, where $D_i$ indicates the $i$-th IIoT device. Specifically, the proposed framework includes the following five steps: 
\begin{itemize}
\item \textbf{Autoencoder Training.} To protect data privacy and reduce communications overhead, each device separately trains an autoencoder in an unsupervised way to ``encrypt" the raw data into latent vectors. The trained encoder functions can extract the critical features into the latent vector from the input data, and ensure dimension reduction of the input data with a compression ratio to improve the wireless transmission efficiency.

\item \textbf{Uploading Decision and Bandwidth Allocation.} The edge server calculates the optimal uploading decision and bandwidth allocation vector according to some critical parameters (such as devices' computation capability, transmission power, channel condition, and latent vector size) and then notifies the devices. 

\item \textbf{Latent Vectors and Labels Uploading.} The devices upload their latent vectors and labels\footnote{We assume that the collected image data has been labeled well with the image labeling tools such as LabelMe~\cite{labelme}, and then correlated to the corresponding latent vectors.} to the edge server with the allocated bandwidth. 

\item \textbf{Edge-CNN Training.} On receiving the latent vectors and labels from the devices, the edge server trains the edge-CNN model until it converges. 

\item \textbf{Trained Edge-CNN Broadcasting and Local Inferece.} The BS broadcasts the well-trained edge-CNN model to all the devices enabling them to perform local inference (i.e., classifying objects from the collected images) independently. 
\end{itemize}

{In this paper, TL functions in the proposed edge-CNN framework in the following two aspects. In particular, when the labeled data is too limited for the edge server to train a CNN from scratch, the edge server uses the existing image dataset to train the CNN in advance, which can be further fine-tuned at the edge server based on the limited datasets uploaded from the devices. On another note, with the aid of TL, the devices that are not participating in the training procedure just need to fine-tune the received edge-CNN model with their local datasets without training from scratch.}

\subsection{Training Principles}
Next, we detail critical training principles including the training on autoencoder and training on the edge-CNN model.

\begin{figure}[t] 
\centering
\captionsetup{font={footnotesize }}
\subfigure[]{
\includegraphics[width=2.25in,height=1.65in]{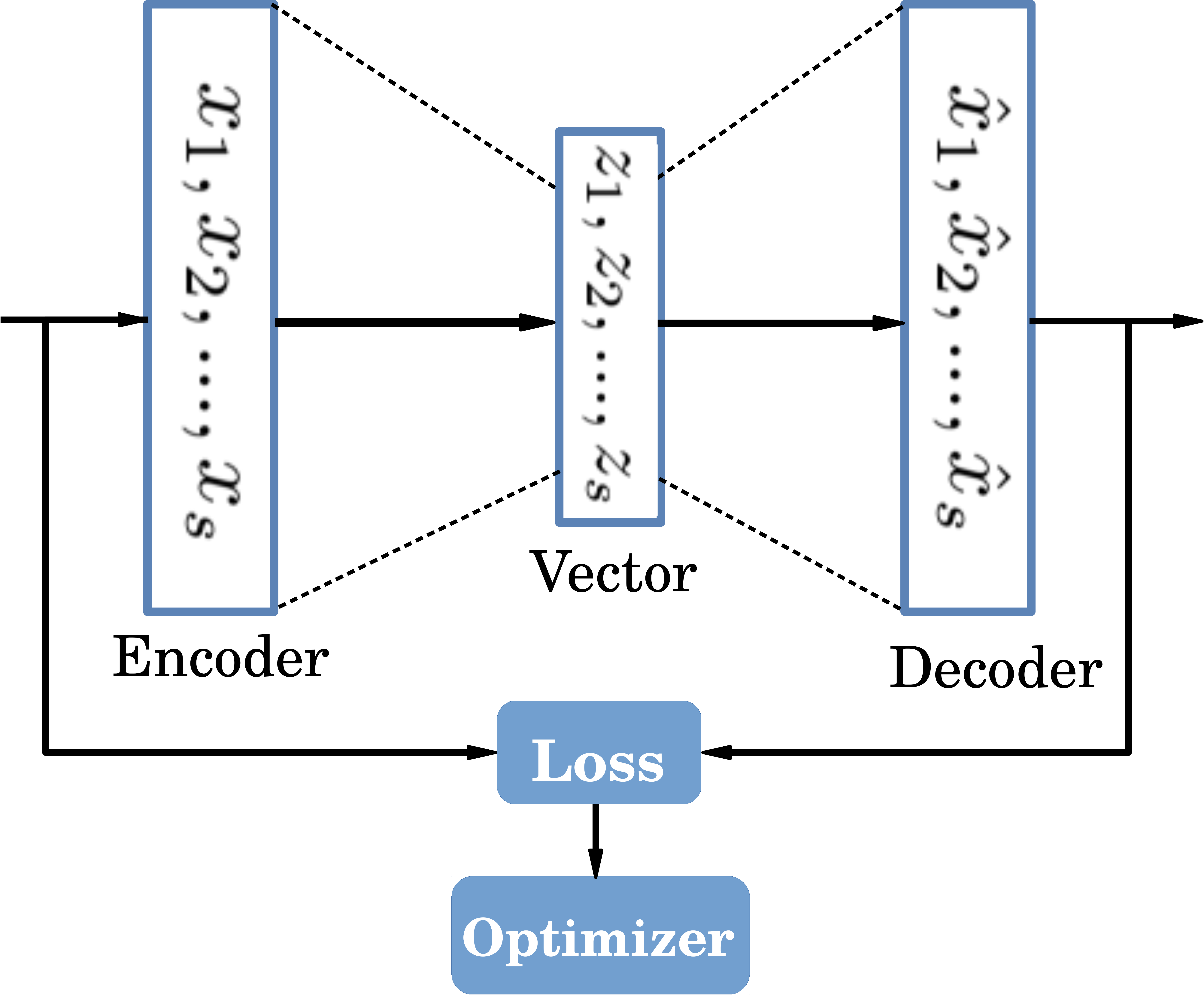}}
\hspace{1.5mm} 
\subfigure[]{
\includegraphics[width=2.65in,height=1.65in]{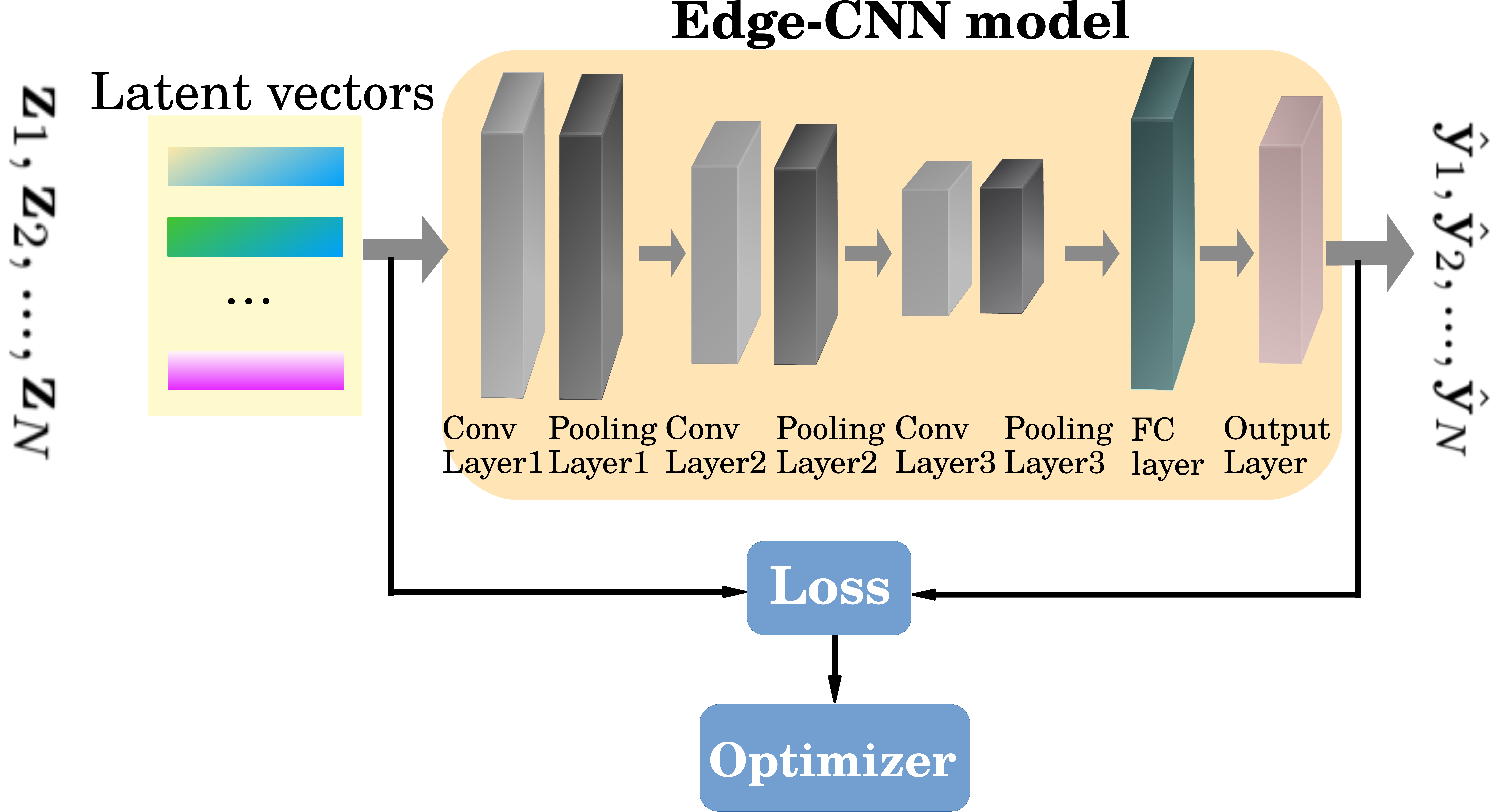}}
\caption{{Training of the autoencoder is shown in (a), where $s$ indicates the number of data samples. The encoder compresses the input data samples $\{x_1, x_2, ..., x_s\}$ into latent vector $\{z_1, z_2, ..., z_s\}$, which can be reconstructed as  $\{\hat{x}_1, \hat{x}_2, ..., \hat{x}_s\}$. Training of the edge-CNN model is shown in (b), where the input is the latent vectors $\{ {\bf z}_1, {\bf z}_2, ..., {\bf z}_N \}$ and the labels are denoted as  $\{ \hat{{\bf y}}_1, \hat{{\bf y}}_2, ..., \hat{{\bf y}}_N \}$.}}
\label{training}
\end{figure}  

\subsubsection{Training on Autoencoder}
The autoencoder deployed at each device is a type of artificial neural network used to learn efficient data codings and reconstruct the original input in an unsupervised way. The components of an autoencoder and its training process is shown in Fig.~\ref{training}(a).
 It is composed of an encoder and a decoder~\cite{autoencoder}. The encoder compresses the input data samples, $\{x_1, x_2, ..., x_s\}$, into a low dimensional representation of pre-determined size, called the latent vector. Suppose that the latent vector is denoted by ${\bf z}=\{z_1, z_2, ..., z_s\}$ that contains the essential features of the data, where $s$ indicates the number of data samples. At the edge server, the decoder tries to reconstruct the original data, $\{\hat{x}_1, \hat{x}_2, ..., \hat{x}_s\}$, from the latent vector ${\bf z}$. {Note that instead of uploading the raw image data, with the well-trained autoencoder, only the compressed latent vector that contains critical features of raw data will be uploaded to the edge server. Since an autoencoder is a lossy network, the raw data may not be fully recovered due to the compression ratio.} Each device trains the autoencoder by minimizing the differences between the original input (i.e., raw image data) and the reconstructed input. 
The autoencoders were trained from scratch using Glorot uniform
method %~\cite{Glorot}
 as initializer, mean square error %~\cite{MSE} 
 as the loss function and RMSprop optimization %~\cite{rmsprop}
  algorithm as the optimizer.

 \begin{figure*}[t] 
\centering
\captionsetup{font={footnotesize }}
\subfigure[]{
\includegraphics[width=2.1in,height=1.68in]{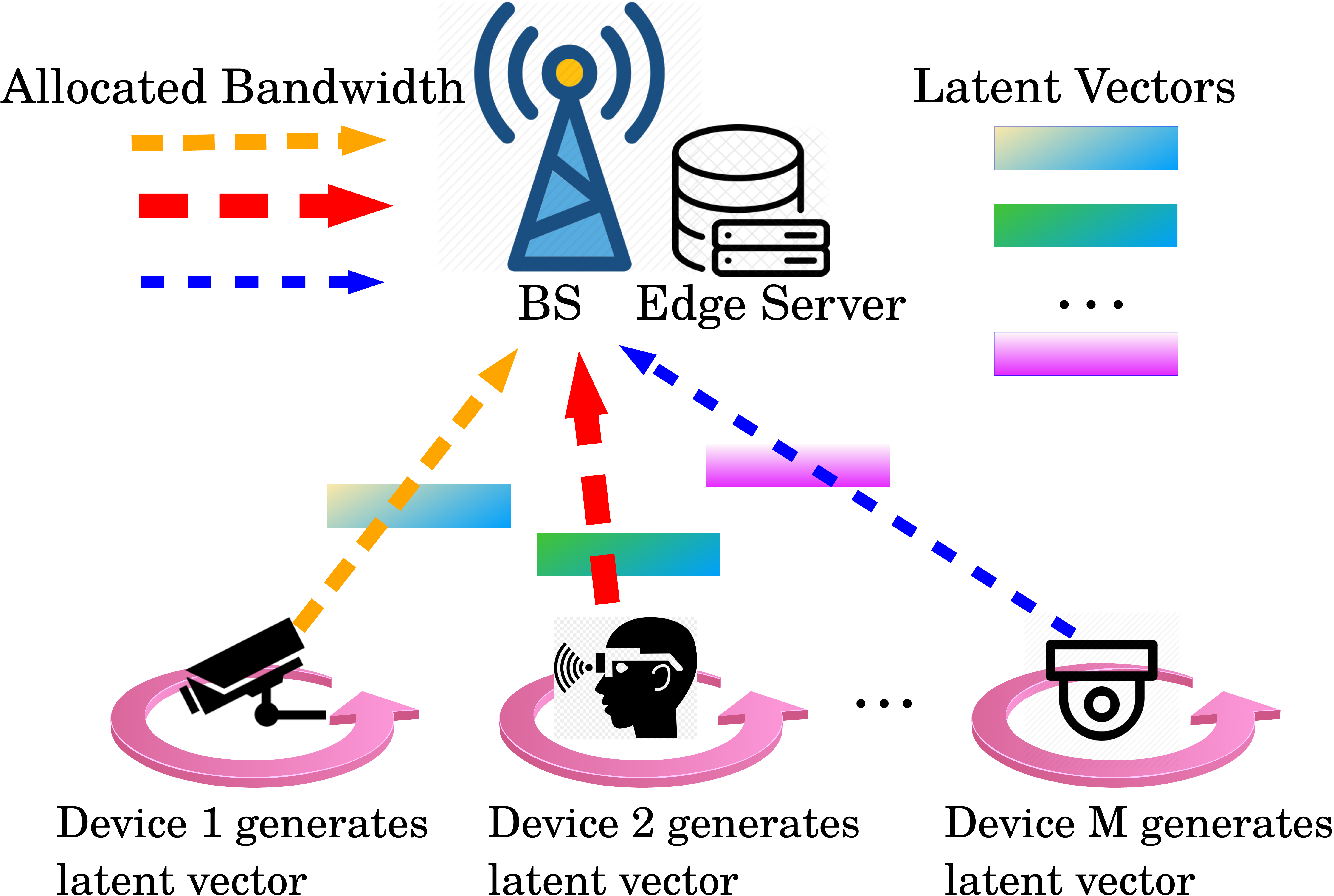}}
\hspace{-0.1in}
\subfigure[]{
\includegraphics[width=2.1in,height=1.68in]{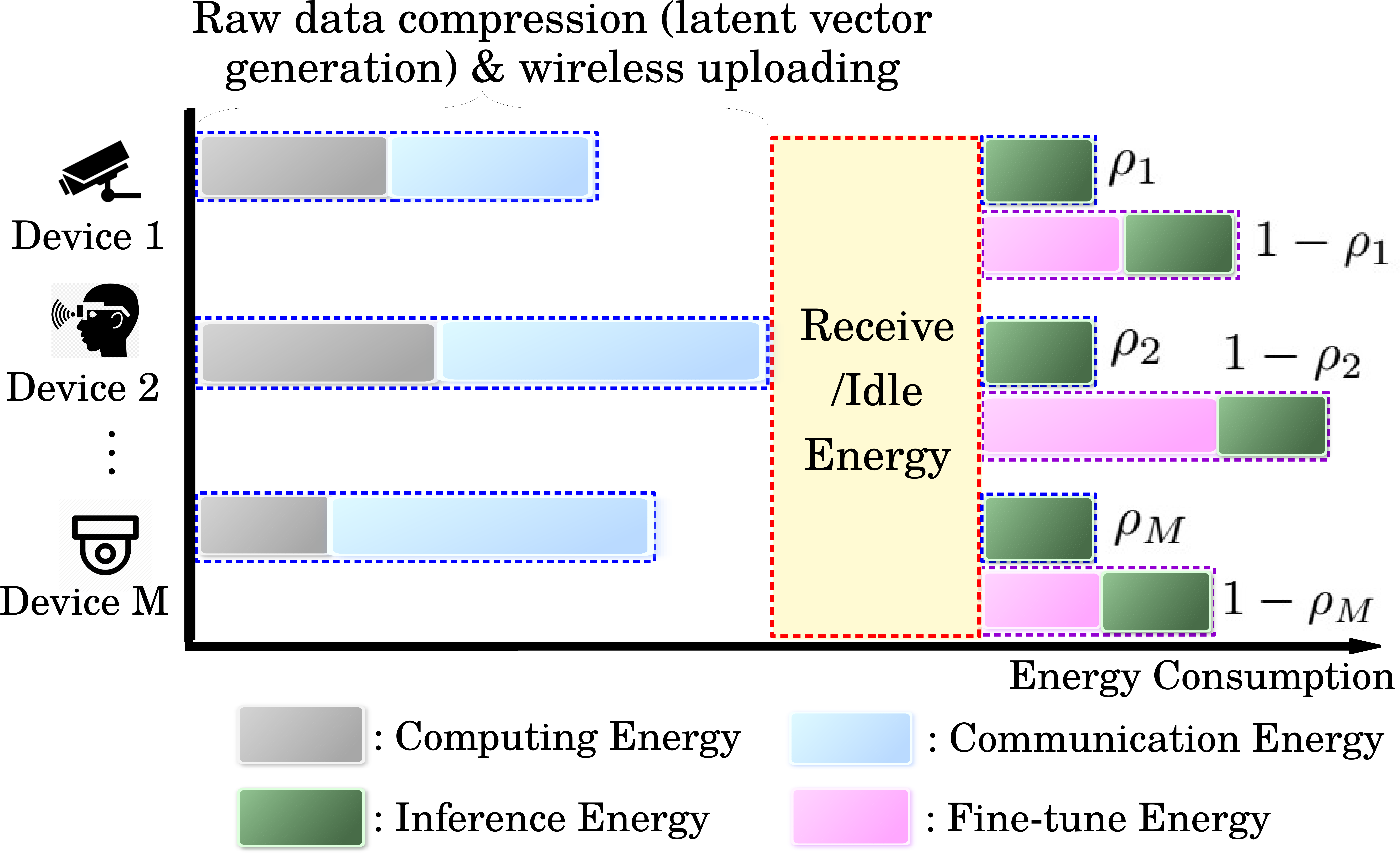}}
\hspace{-0.05in}
\subfigure[]{
\includegraphics[width=2.1in,height=1.68in]{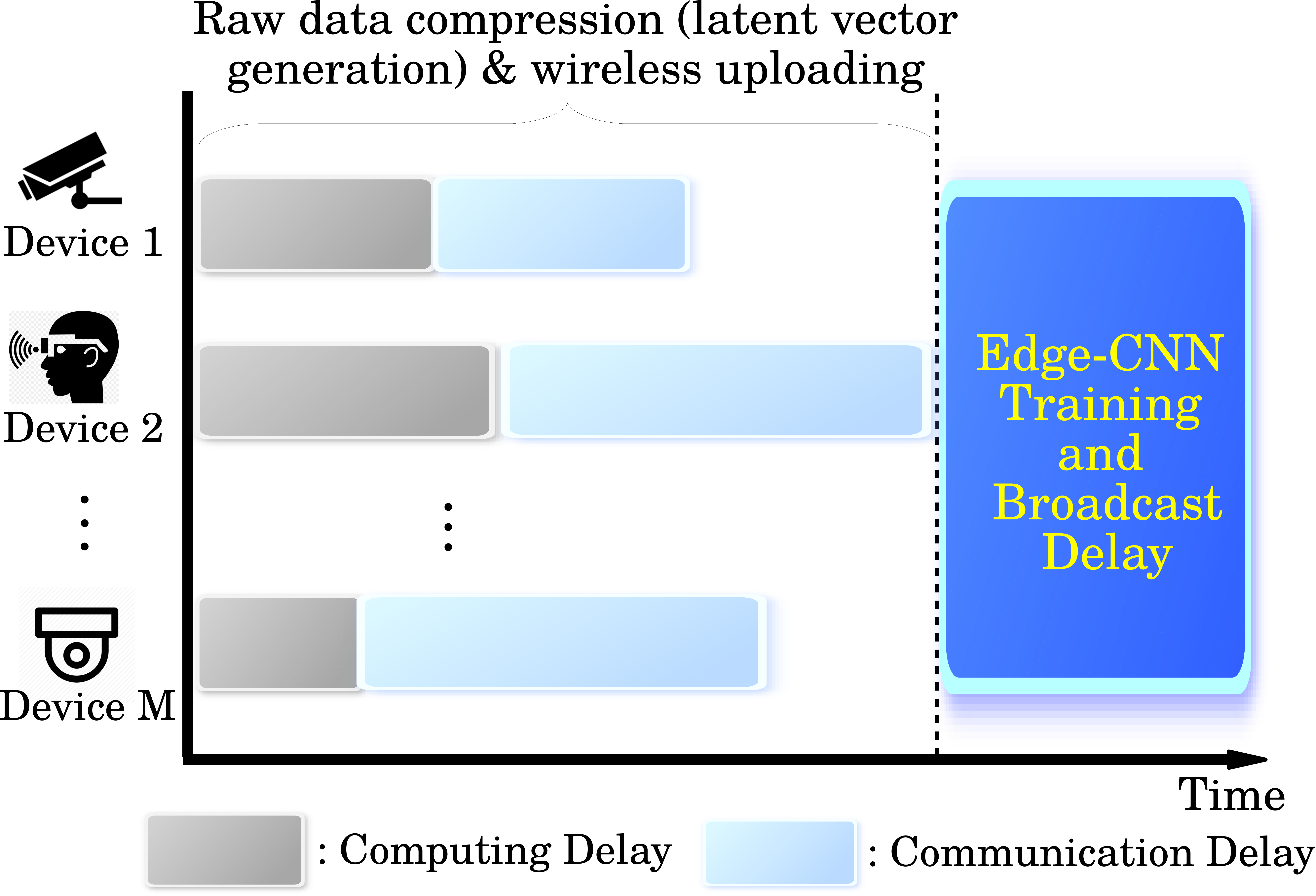}}
\hspace{-5mm}
\caption{{An illustrative example of the proposed framework for latent vectors and labels uploading via the allocated bandwidth is shown in (a), where the thickness of the arrow indicated the allocated bandwidth. Energy consumption achieving image classification of each device is illustrated in (b), where we consider the non-iid data setting and the $i$-th device participates in the edge-CNN training with a binary decision $\rho_i$. Involved computing and communication latency before the edge-CNN training is illustrated in (c), where the edge-CNN training process is susceptible to the ``straggler effect".}} %where there exists a tradeoff between the computing latency and communication latency.}}
\label{model}
\end{figure*} 
 
{After the autoencoder is well-trained, the encoder part is extracted and deployed on the device to infer the latent vector accordingly.  Denote the trained autoencoder of the $i$-th device as ${\cal A}_i$, which is able to extract critical features in the raw data samples and encrypt the data by transforming the raw data into latent vectors ${\bf z}$ to upload to the edge server for privacy-preserving. This is because an adversary could not reconstruct the raw data samples from the latent vector without knowing the structure (e.g., number of layers and number of nodes in each layer) and all the weights of the trained autoencoder (e.g., the decoder part). Even if the device is captured by an adversary, it is still very difficult for the adversary to deduce the decoder part from the encoder part deployed  on the device.}

\subsubsection{Training on Edge-CNN}
In this stage, the locally-trained encoder is extracted and deployed at each device in the inference mode. For the $i$-th device, the collected dataset ${\bf x}_i=\{x_1, x_2, ..., x_s\}$ is fed to the corresponding pre-trained encoder, which transforms the input data ${\bf x}_i$ into a latent vector ${\bf z}_i$ which represents the critical features of ${\bf x}_i$. In this context, the edge-CNN model is trained based on the latent vectors (i.e., the input to the edge-CNN) and the corresponding labels ${\bf y}_i$ (i.e., the output of the edge-CNN), which are uploaded by the devices, as illustrated in Fig.~\ref{training}(b). Specifically, we conventionally use the back-propagation and set the loss function of the CNN classification as cross-entropy. %\cite{bce}.
 We use Adams optimization algorithm as the
optimizer and data augmentation was further used during the training process to mitigate overfitting.  The probability of each output class is predicted using the Softmax function, and the predicted probability for the $c$-th class is given as
%\begin{equation}\label{softmax}
$\sigma_c(k) =\frac{e^{k_c}}{\sum_{j=1}^{M} e^{k_j}}$, where $ c=1,...,C,$ %\end{equation}
 $C$ is the total number of classes, $k$ is the output of the last fully connected layer.
 During the training on edge-CNN, the key design metric is the number of devices transferring their latent vectors to the edge server and how to allocate bandwidth to them, as highlighted in \textbf{Remark~\ref{R1}}.

\theoremstyle{Remark}
\newtheorem{remark}{\textbf{Remark}}
\begin{remark} 
\label{R1}
Since each device may have unique training data samples which are usually non-iid, the edge server generally prefers to include more latent vectors from devices. On another note, due to the energy constraint,                                                                                                                                                                                                                                                                                                                                                                                                                                                                                                                                                                                                                                                                                                                                                                                                                                                                                                                                                                                                                                                                                                                                                                                                                                                                                                                                                                                                                                                                                                                                                                                                                                                                                                                                                                                                                                                                                                                                                                                                                                                                                                                                                                                                                                                                                                                                                                                                                                                                                                                                                                                                                                                                                                                                                                                                                                                                                                                                                                                                                                                                                                                                                                                                                                                                                                                                                                                                                                                                                                                                                                                                                                                                                           only a subset of devices may upload their latent vectors to the BS. Considering the various computation capability of each device and the channel conditions, the BS should allocate appropriate wireless bandwidth for these devices to minimize the completion latency. 
\end{remark}

\subsection{Energy Model}
To illustrate the device's energy consumption, we first define the uploading decision as follows.
\theoremstyle{Definition}
\newtheorem{definition}{\textbf{Definition}}
\begin{definition} 
\label{D1}
\textbf{Uploading Decision:} The uploading decision of $D_i$ (denoted as $\rho_i \in \{0, 1\}$) is defined as the \textbf{binary} decision indicating whether $D_i$ participates in the edge-CNN training or not, i.e., $D_i$ uploads its latent vector and labels to the BS when $\rho_i=1$, and vice versa. 
\end{definition} 

Based on \textbf{Definition~\ref{D1}}, we assume that the devices upload the latent vectors and labels to the BS via orthogonal-access schemes such as orthogonal frequency division multiple access (OFDMA), as illustrated in Fig.~\ref{model}(a).

\subsubsection{Energy Consumption with Vectors Uploading}
In this case, the device's energy consumption includes five parts: ``local data encoding", ``latent vector transmission", ``waiting for edge-CNN training", ``trained edge-CNN model downloading", and ``local inference", as illustrated in Fig.~\ref{model}(b). In this case, the total energy consumption of $D_i$ with vectors uploading is calculated as
\begin{equation}\label{energy_upload}
E^i_{U^+} = e^i_{encode} + e^i_{trans} + e^i_{idle} + e^i_{down} + e^i_{inf}, \ \ \forall i \in \cal M.
\end{equation}

\subsubsection{Energy Consumption without Vectors Uploading}
Unlike the previous case, the device's energy consumption in this case also includes five parts: ``local data encoding", ``waiting for edge-CNN training", ``trained edge-CNN model downloading", ``edge-CNN model fine-tuning", and ``local inference". In this case, the total energy consumption of $D_i$ without vectors uploading is
\begin{equation}\label{energy_w/o_upload}
E^i_{U^-} = e^i_{encode} + e^i_{idle} + e^i_{down} + e^i_{tune} + e^i_{inf}, \ \ \forall i \in \cal M.
\end{equation}

To this end, total energy consumption of $D_i$ is obtained as
\begin{equation}\label{energy_total}
E_i = \rho_i E^i_{U^+} + (1-\rho_i) E^i_{U^-}.
\end{equation}

%\subsubsection{Energy Minimization Problem and Optimal Uploading decision}
Substituting (\ref{energy_upload}) and (\ref{energy_w/o_upload}) into (\ref{energy_total}), we have
\begin{equation}\label{energy_total_new}
E_i = \rho_i \Delta_e^i + E^i_{U^-},
\end{equation}
where $\Delta_e^i =  e^i_{trans}  - e^i_{tune}$.

\subsection{Latency Model}
Here, we analyze the computing latency and communication latency before the edge-CNN learning, as illustrated in Fig.~\ref{model}(c). We denote the total number of devices participating the edge-CNN training as $N= \sum_{i=1}^{M} \rho_i$, $N \leq M$, and the devices set is given by $\cal N$. 
%\subsection{Wireless Bandwidth Allocation Optimization with Minimized Computing and Communication Latency}
\subsubsection{Local Computing Latency}
For the local computing, each device $D_i, \ \forall i \in \cal N$, generates the latent vector via the trained autoencoder ${\cal A}_i$ with their local datasets ${\bf x}_i$. For such autoencoder inferring, the computation delay (also local inferring delay) mainly depends on the device's computing capability and the size of data samples. Denote the number of computation operations to calculate the latent vector with respect to one data sample (also one image) as $c$, the number of data samples loaded by device $D_i$ as $|{\cal D}_i|$, and the CPU frequency of $D_i$ as $f_i$, and accordingly the delay cost of local autoencoder training of $D_i$ is calculated as
\begin{equation}\label{t_local}
t^i_{comp} ={|{\cal D}_i| c}/{f_i}, \ \ \forall i \in \cal N.
\end{equation}

\subsubsection{Wireless Transmission Latency}
After the latent vectors are generated, we assume that the devices upload their vectors and labels to the BS via OFDMA in an asynchronous manner. In the communication model, suppose that $B$ indicates the total bandwidth between the devices and the BS. $P_i$ is the transmission power of $D_i$. %, which is assumed to be identical for all devices. 
$h_i$ is the channel gain between $D_i$ and the BS. $N_0$ is the Gaussian noise variance. %In is the interference caused by the users that connect to other BSs using the same RB. 
The achievable rate of device $D_i$ can be obtained as
\begin{equation}\label{r_commu}
r^i_{comm}(\mathbf{w}) = {w}_i B {\rm log}_2 \left(1+ \frac{P_i h_i}{ N_0}\right) , \ \forall i \in \cal N,
\end{equation}
where $\mathbf{w}=[{w}_1, {w}_2, ..., {w}_N]$ is the bandwidth allocation ratio vector, ${w}_i \in [0,1]$ indicates the ratio of the wireless bandwidth allocated to $D_i$.
Due to limited bandwidth of the system, we have $0 \leq {w}_{i} \leq 1$ and $\sum_{i=1}^{N} {w}_{i} = 1$.
 
Denotes the size of each data sample of $D_i$ as $s_i$, and the compression ratio is $ \lambda_i $, then the size of the latent vector of each data sample can be obtained as $s_i/\lambda_i$. %$(1-\lambda)s_i$.
 Therefore, the total size of the output latent vector of $D_i$ is given by $v_i = \frac{|{\cal D}_i| s_i}{\lambda_i}$. Given the uplink data rate in (\ref{r_commu}), the transmission delay between $D_i$ and the BS over uplink is specified as
\begin{equation}\label{t_commu}
t^i_{comm}(\mathbf{w}) = \frac{v_i}{r^i_{comm}(\mathbf{w})} = \frac{v_i}{{w}_i B {\rm log}_2 \left(1+ \frac{P_i h_i}{ N_0}\right)}, 
\end{equation}
where $v_i$ is the data size (number of bits) of the latent vectors transmitted from $D_i$ to the BS over wireless link\footnote{The size of label is ignored in this paper since the label size is relatively small compared to latent vector.}.
For a given compression ratio, we envision that the latent vectors of each device is similar to each other. Hence we assume $v_i=v$.
 
On receiving the latent vectors and labels from the devices, the edge server trains the edge-CNN model accordingly. Hereinafter, the total latency needed before the training of edge-CNN is called \textit{completion latency}, which includes the local computation time and the wireless transmission time, as illustrated in Fig.~\ref{model}(b). Based on (\ref{t_local}) and (\ref{t_commu}), the completion latency of the proposed framework is
\begin{equation}\label{latency}
\begin{aligned}
T(\mathbf{w}) \! &=\! \underset{i \in {\cal N}}{\rm max} \left \{t^i_{comp} +  t^i_{comm}(\mathbf{w}) \right \} \\
\! &=\!\underset{i \in {\cal N}}{\rm max} \left \{\frac{|{\cal D}_i| c}{f_i} \!+\! \frac{v}{{w}_i B {\rm log}_2 \left(1\!+\! \frac{P_i h_i}{ N_0}\right)} \right \}.
\end{aligned}
\end{equation}

It is observed from (\ref{latency}) that the slowest end device will fundamentally limit the total latency $T(\mathbf{w})$ in the proposed framework. This leads to the so-called \textit{straggler’s effect issue}, as detailed in \textbf{Remark~\ref{R2}}.

\begin{remark}
\label{R2}
The edge-CNN training process of the proposed framework is susceptible to the ``straggler effect" issue. The edge-CNN training only progresses as fast as the devices with the slowest computation and communication since the edge server waits for the latent vectors from all devices before edge-CNN training can take place, {as illustrated in Fig.~\ref{model}(c).}
\end{remark}

\section{Joint Latent Vector Uploading Decision and Wireless Bandwidth Allocation}
\label{optimization}
In this section, we formulate the Joint Uploading Decision and Bandwidth Allocation (JUDBA) problem as a weighted-sum cost minimization problem, followed by the problem decomposition and solutions to the JUDBA problem.

\subsection{Problem Formulation}
                                                                                                                                                                                                                                                                                                                                                                                                                                                                                                                                                                                                                                                                                                                                                                                                                                                                                                                                                                                                                                                                                                                                                                                                                                                                                                                                                                                                                                                                                                                                                                                                                                                                                                                                                                                                                                                                                                                                                                                                                                                                                                                                                                                                                                                                                                                                                                                                                                                                                                                                                                                                                                                                                                                                                                                                                                                                                                                                                                                                                                                                                                                                                                                                                                                                                                                                                                                                                                                                                                                                                                                                                                                                                                                 For a given uploading decision $\rho_i$ and uplink bandwidth allocation vector $\bf w$, we define the weighted-sum cost of the $i$-th device ($D_i$) during the edge-CNN training as
\begin{equation}\label{cost}
{\cal O}_i=\alpha_i E_i + (1-\alpha_i) T(\mathbf{w}), \ \ \forall i \in {\cal N},
\end{equation}
in which $\alpha_i$, $0 \le \alpha_i \le 1$, specifies the $i$-th device's preference on energy consumption, and $1-\alpha_i$ specifies the preference on completion time. For example, an IIoT device $D_i$ with short battery life can increase the coefficient $\alpha_i$ so as to save energy at the expense of latency, and vice versa.

Let ${\bf Q} = \{\rho_1, \rho_2, ..., \rho_M \}$ denote the uploading decision vector, we formulate the JUDBA problem as a weighted-sum cost minimization problem, i.e.,                                                                                                                                                                                                                                                                                                                                                                                                                                                                                                                                                                                                                                                                                                                                                                                                                                                                                                                                                                                                                                                                                                                                                                                                                                                                                                                                                                                                                                                                                                                                                                                                                                                                                                                                                                                                                                                                                                                                                                                                                                                                                                                                                                                                                                                                                                                                                                                                                                                                                                                                                                                                                                                                                                                                                                                                                                                                                                                                                                                                                                                                                                                                                                                                                                                                                                                                                                                                                                                                                                                                                                                                                                                                                                                                                                                                                                                                                                                                                                                                                                                                                                                   
\begin{subequations} \label{problem}
\begin{align}
 \rm  & ({\bf P1}):  \notag  \ \underset{\{\mathbf{Q, w}\}}{\rm min}  {\cal O}_i \notag  \\
& {\rm{s}}{\rm{.t}}{\rm{.}}\;\;\;{\mathbf{C1}: \; \sum_{i=1}^{N} {w}_{i} \leq 1,}\\
& \;\;\;\;\;\;\;\;\mathbf{C2}: \; 0 \leq {w}_{i} \leq 1, \\
&  \;\;\;\;\;\;\;\;\mathbf{C3}: \; {\rho}_{i} \in \{0, 1 \},
\\
&  \;\;\;\;\;\;\;\;\mathbf{C4}: \; N \leq M.
\end{align}
\end{subequations}

The constraints in ${\bf P1}$ is explained as
follows: {constraint (\ref{problem}a) indicates that the allocated bandwidth can not exceed the total bandwidth due to the attenuation or losses in practice}. Constraint (\ref{problem}b) gives the range of the bandwidth allocation ratio of each device. Constraint (\ref{problem}c) gives the binary choice of the uploading decision. At last, constraint (\ref{problem}d) indicates that the number of devices uploading the vectors is no larger than the total number of devices.

\subsection{Problem Decomposition}
Through exploiting the structure of the objective function and constraints (i.e., $\mathbf{C1}$-$\mathbf{C4}$) of JUDBA problem, we observe that the JUDBA problem with high complexity can be decomposed into %an equivalent \textit{master problem} and a \textit{slave subproblem} 
two subproblems with separated objective and constraints with the Tammer decomposition method \cite{Tammer}.  Therefore, we first rewrite the JUDBA problem (${\bf P1}$) as
\begin{equation} \label{problem1}
\begin{aligned}
 \rm  (\widetilde{\bf P1}):  & \ \underset{\bf Q}{\rm min} \ \left(\underset{\mathbf{w}}{\rm min} \ {\cal O}_i \right)   \\
&{\rm{s}}{\rm{.t}}{\rm{.}}\;\;\;\mathbf{C1-C4}.
\end{aligned}
\end{equation}

To solve the equivalent problem $\widetilde{\bf P1}$, we decompose  $\widetilde{\bf P1}$ into two subproblems, as highlighted in \textbf{Remark~\ref{R3}}.
\begin{remark} 
\label{R3}
Solving $\widetilde{\bf P1}$ is equivalent to solving two subproblems without changing the optimality\cite{MINLP_01}:
i) Bandwidth allocation subproblem ($\mathbf{P1.1}$) with fixing the uploading decision to trace the latency-learning tradeoff by minimizing the completion time, i.e.,  \begin{equation} \label{BA}
\begin{aligned}
& (\mathbf{P1.1}): \ {\cal O}^*_{i}=\underset{\bf w }{\rm min}\; {\cal O}_{i}    \\
&{\rm{s}}{\rm{.t}}{\rm{.}}\;\;\;\mathbf{C1, C2},
 \end{aligned}
\end{equation}

 ii) Uploading decision subproblem ($\mathbf{P1.2}$) with the optimal bandwidth allocation to trace the energy-learning tradeoff by minimizing the device's energy consumption, i.e., \begin{equation} \label{UP}
\begin{aligned}
& (\mathbf{P1.2}): \ \underset{\bf Q}{\rm min}\;  {\cal O}^*_{i}   \\
& {\rm{s}}{\rm{.t}}{\rm{.}}\;\;\;\mathbf{C3, C4}.
 \end{aligned}
\end{equation}
 
\end{remark}
\subsection{Problem Solutions}
Here, we present our solutions to subproblems $\mathbf{P1.1}$ and $\mathbf{P1.2}$ so as to obtain the solution to the JUDBA problem.

\subsubsection{Solution to $\mathbf{P1.1}$}
Firstly, given a feasible uploading decision vector ${\bf Q}'=\{\rho_1', \rho_2', ..., \rho_M' \}$ that satisfies constraints (\ref{problem}c) and (\ref{problem}d), then the corresponding number of devices uploading the latent vectors can be achieved as $N= \sum_{i=1}^{M} \rho_i'$.
 Here, we introduce {\textbf{Theorem~\ref{Tm1}}} to solve the latency minimization problem.
 
\theoremstyle{Theorem}
\newtheorem{theorem}{\textbf{Theorem}}
\begin{theorem}
\label{Tm1}
The solution to $\mathbf{P1.1}$ is calculated as follows:
\begin{equation}\label{solution}
\begin{aligned}
{w}_i^* & = {\rm arg \underset{\mathbf{w}}{\rm min}} \ T(\mathbf{w}) \\
& = \frac{v}{ B {\rm log}_2 \left(1\!+\! \frac{P_i h_i}{ N_0}\right) \left(T^*-\frac{|{\cal D}_i| c}{f_i}\right)},
\end{aligned}
\end{equation}
where $T^*$ denotes the minimal completion latency that satisfies the following condition:
\begin{equation}\label{condition1}
{\sum_{i=1}^{N} \frac{v}{ B {\rm log}_2 \left(1\!+\! \frac{P_i h_i}{ N_0}\right) \left(T^*- \frac{|{\cal D}_i| c}{f_i}\right)} \leq 1.}
\end{equation}
\end{theorem} 

\begin{proof}

In the proposed framework, if some device generates the latent vector slower than others, the BS could allocate more wireless bandwidth for these slower devices to accelerate their transmission procedure. As a result, the system completion latency can be shortened. %Also, it is observed that $\mathbf{P1.2}$ is a linear program, which can be solved using the Karush-Kuhn-Tucker (KKT) conditions~\cite{KKT}. 
Therefore, the optimal solution can be achieved only when all devices finish the raw data compression (i.e., compress the raw images into the latent vectors) and wireless uploading at the same time, i.e., $T(\mathbf{w})=T^*$. 

Substituting (\ref{latency}) into $T(\mathbf{w})=T^*$, we have
\begin{equation} \label{problem1}
 \frac{|{\cal D}_i| c}{f_i}+ \frac{v}{{w}_i^* B {\rm log}_2 \left(1+ \frac{P_i h_i}{ N_0}\right)} = T^*,  
\end{equation}
where $T^*$ meets the condition (\ref{condition1}) and can be achieved using the bisection search, a.k.a. interval halving or binary search~\cite{binary_search}.
As a result, ${w}_i^*$, can be derived as (\ref{solution}).
\end{proof}

To capture the features of ${w}_i^*$, \textbf{Observation~\ref{p1}} is given as follows.

\theoremstyle{Observation}
\newtheorem{observation}{\textbf{Observation}}
\begin{observation}
\label{p1} We observe from (\ref{solution}) that ${w}_i^*$ is inversely proportional to $h_i$ and $f_i$. This indicates that the device with worse channel conditions or weaker computation capabilities will be allocated with more bandwidth. 
\end{observation}

\subsubsection{Solution to $\mathbf{P1.2}$}
To achieve $\rho_i^*$, exhaustive search method can be used to explore over all possible uploading decisions. Also, we can introduce heuristic algorithm to find the optimal solution (sometimes local optimal) more efficiently~\cite{MINLP_01}. In this context, we have {\textbf{Observation~\ref{O1}}}.

\begin{observation} 
\label{O1}
Given the $w_i^*$, then $\rho_i^*$ is determined by $\Delta_e^i$. If $\Delta_e^i>0$, i.e., the energy consumption of latent vectors uploading dominates, so we have $\rho_{i}^*=0$ since $E_i$ increases monotonously with $\rho_i$, and vice versa. 
\end{observation} 

%\subsection{Implementation Issues}
To this end, ${\bf Q}^* = \{\rho_1^*, \rho_2^*, ..., \rho_M^* \}$,  the optimal solution for the JUDBA problem is achieved as $\{{\bf Q}^*, {\bf w}^* \}$, where ${\bf w}^*$ is obtained in (\ref{solution}) by setting $\rho_i=\rho_i^*$. Note that $\mathbf{P1}$ is a mixed-integer nonlinear programming (MINLP) problem and solving $\bf P1$ usually converges slowly and requires forbidding complexity for \textit{real-time implementation} in the IIoT-densely-deployment scenario since too many iterations are required.
 Interestingly, we can explore the advantages of {deep learning} to decrease the involved complexity for the ease of implementation in practice~\cite{YB_TMC}. For example, with the offline-trained \textit{multi-task learning} model, the BS can directly infer $\{{\bf Q}^*, {\bf w}^* \}$ together with high accuracy through feedforward calculation without iterations.
 
\section{Experiments Results}
\label{result}
In this section, we conduct experiments to demonstrate the CNN learning performance under different compression ratio. Then we evaluate the performance of the proposed framework.% and trace the energy-latency-learning tradeoffs.

\subsection{CNN Training and Inference}
In the experiments, we evaluate the vanilla CNN model performance with the ImageNet dataset~\cite{GoogleNet}, which contains about 13,000 images each resized to a dimension 256$\times$256$\times$3, spanning 10 classes, and the training testing ratio is 7:3. The training and testing of the autoencoders and CNN classifier were implemented using Keras on TensorFlow backend, running on a NVIDIA Tesla P100-PCIE-16GB GPU.

\begin{table}[thb]
\centering
  \captionsetup{font={footnotesize}} 
\caption{{CNN performance versus compression ratio on ImageNet dataset}}
 \begin{tabular}{|*{6}{c|} }
 \hline
  %\multicolumn{6}{c|}{224*224*3}   \\  
 %\hline
 Compression & Inference & Testing & Training & Model \\  
 
 ratio ($\lambda$) &  accuracy (\%) & time (s) & time (s) & parameters ($10^3$)\\  
 \hline\hline
 1 & 83 & 1.92  & 21848.81 & 2798.25\\  

 \hline
4 & 77 & 0.63  & 6249.24 & 619.18 \\  
 
 \hline
 8 & 75 & 0.42  & 3555.93 & 272.11  \\ 
 
 \hline
 16 & 74 & 0.33  & 2208.99 & 131.75  \\  
  \hline
{32} & {69} & {0.33} & {2272.93} & {34.31}  \\ 
 \hline
 {64} & {64} & {0.32} & {1613.18} & {33.45}  \\ 
  \hline
\end{tabular}
\label{training_result}
\end{table}

The experiment results are presented in Table~\ref{training_result}, where the data in the table for compression ratio ``$1$" represents the performance obtained when the \textit{uncompressed} ImageNet images were used to train and test the CNN model, denoted as the `baseline' that other cases are compared to. {It can be observed that for a compression ratio of $32$, the achieved prediction accuracy is up to almost $85\%$ of the baseline while the size of model parameters is only $1\%$ of the baseline.} 
This indicates that with the autoencoder, the training and testing overhead of the CNN model in the proposed edge-CNN framework can be significantly reduced without much degradation on the inference accuracy. 
\subsection{Performance Evaluation} 
%\vspace{2.0mm}

In this part, numerical results are provided to demonstrate the performance of the proposed edge-CNN framework. 
In the simulations, some critical parameters are given in Table~\ref{Parameters}. Without loss of generality, the simulation is performed for an average of $10^3$ runs. 
\begin{table}[thb]  
  \captionsetup{font={footnotesize}} 
%\caption{\\ \scshape Critical Parameters and Values } 
\caption{ Critical parameters and values } 
\label{Parameters}
\centering  
\begin{tabular}{|l | l|}  
\hline
{Parameters} & {Value} \\
\hline  \hline
Channel bandwidth ($B$) & $10$ MHz \\
\hline 
 CPU frequency of device ($f$) &  $0.1 \sim 1.0$ GHz \\
 \hline
 Device's transmission power ($P_t$) & $0.3$ W \\
  \hline 
  Device's idle power ($P_I$) & $0.1$ W \\
 \hline
 White noise power ($N_0$) & $7.9\times10^{-13}$ \\
  \hline
  Image data size ($s$) & $800$ kbits \\
    \hline
  Preference coefficient on energy consumption ($\alpha$) & $0.5$ \\
  \hline
\end{tabular}  
\end{table}  

\subsubsection{Weighted-sum System Cost}
The attainable performance is characterized both by the total weighted-sum cost defined in (\ref{cost}). The performance metrics are evaluated for our proposed framework  with two benchmarks:
\begin{itemize}
\item ``\textit{Fully-uploading}": All the devices participate in the edge-CNN training (i.e., $\rho_i=1, \ \forall i \in \cal M$) and the whole bandwidth is allocated equally to the devices.
\item ``\textit{Randomly-uploading}": Each device is with a random value of $\rho_i$  (i.e., $\rho_i=\{0,1\}, \ \forall i \in \cal M$) and the bandwidth is allocated equally to the devices.
\end{itemize}   

\begin{figure}[t] 
\centering
\captionsetup{font={footnotesize }}
\subfigure[]{
\includegraphics[width=2.35in,height=1.88in]{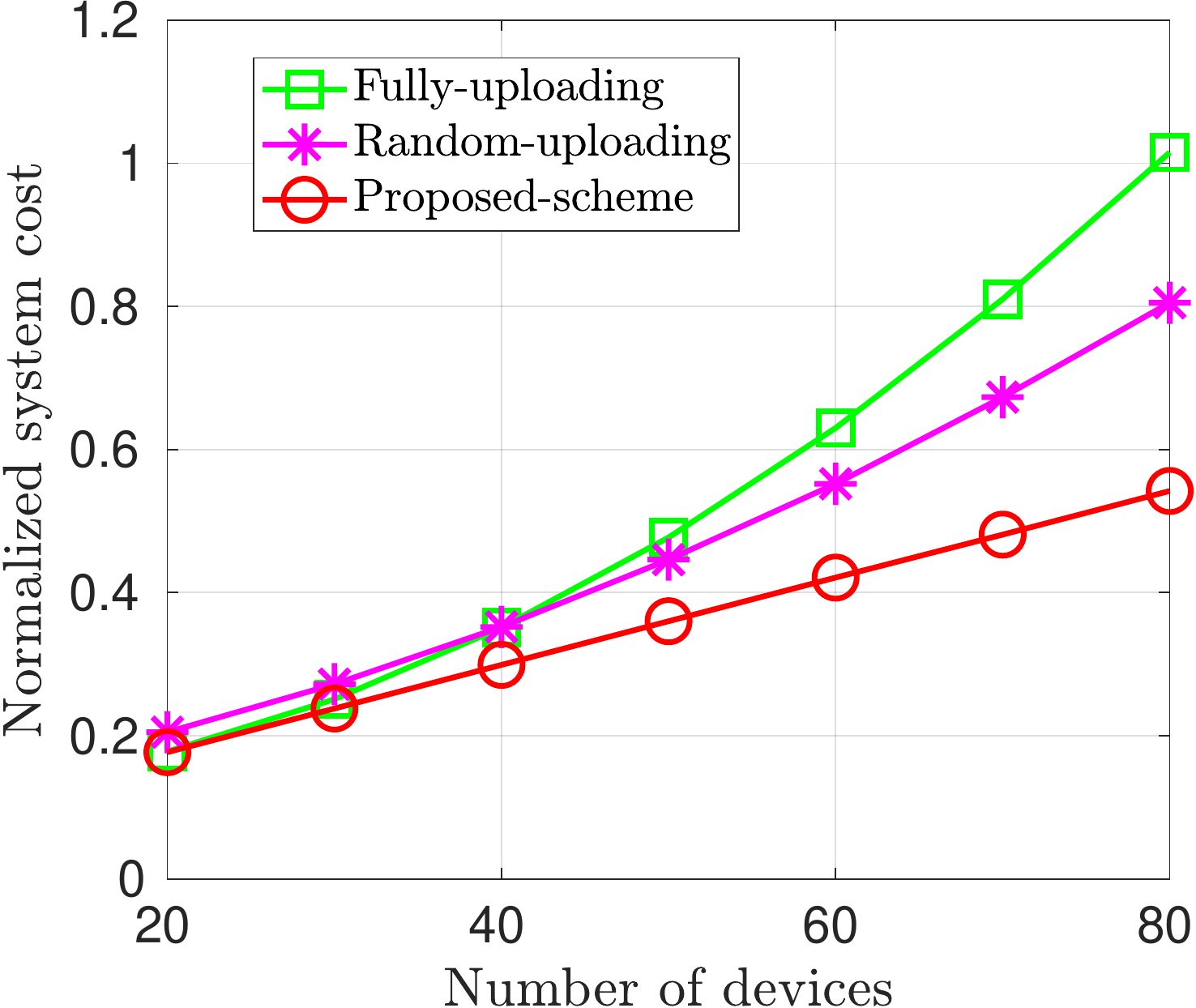}} 
\hspace{0.5in}
\subfigure[]{
\includegraphics[width=2.35in,height=1.88in]{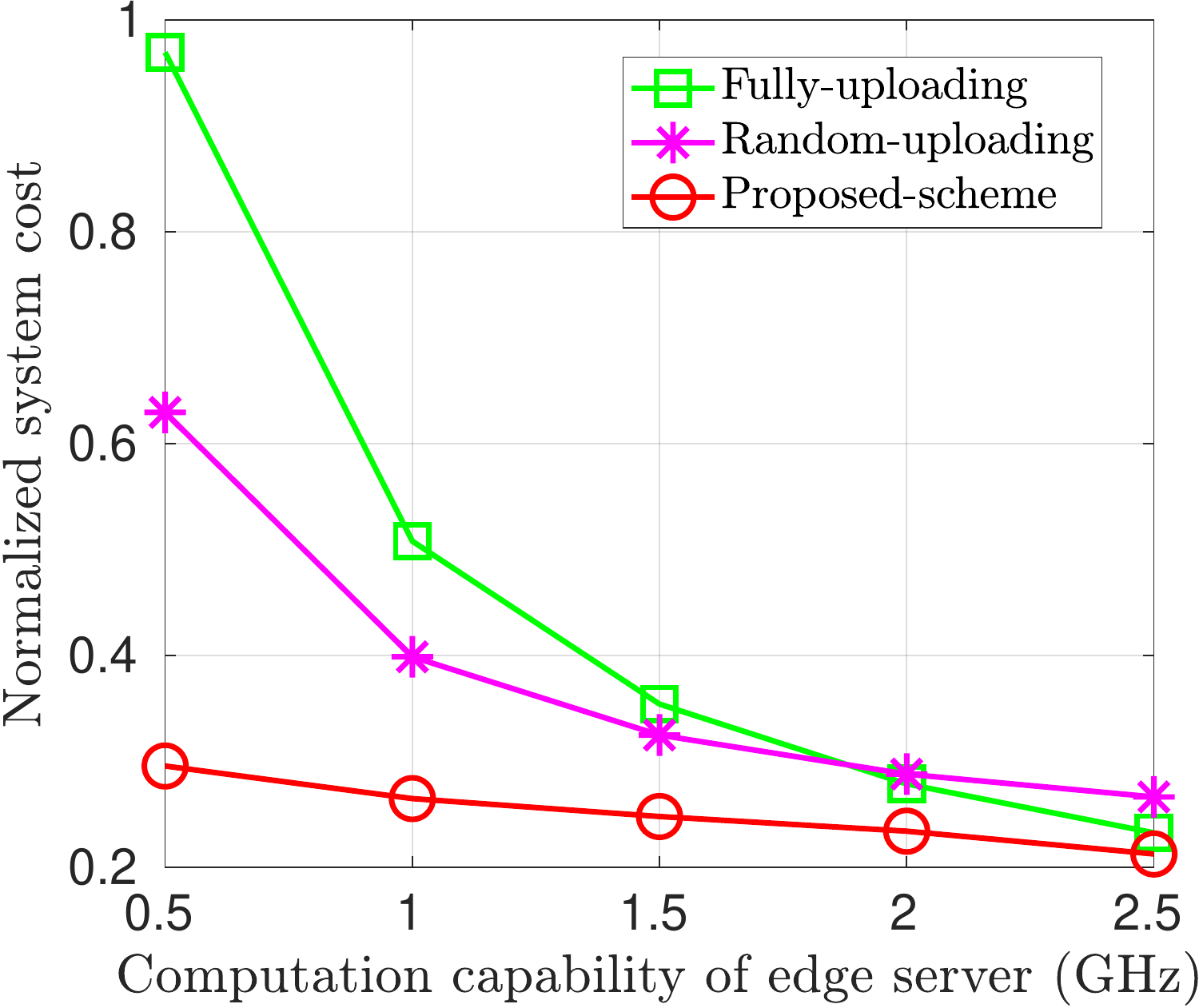}}
\caption{{The normalized weighted-sum cost (defined in (\ref{cost})) versus the number of devices ($N$) is shown in (a). The normalized system cost versus the edge server's computation capability ($F$) is shown in (b), where the total number of devices is $30$. The compression ratio ($\lambda$) is $4$ in (a) and (b).}}
\label{cost}
\end{figure}

Fig.~\ref{cost}(a) presents the impact of the total number of devices ($M$) on the normalized system cost, where $F=2.5$ GHz. We observe that the system cost of three schemes increases by $M$. For example, when $M \le 40$, the ``Fully-uploading" outperforms ``Randomly-uploading" while an opposite trend appears when $M > 40$. This is because the devices are prone to be allocated more bandwidth to participate in the edge-CNN training when $M$ is small. As $M$ keeps increasing, the limited bandwidth is not sufficient to meet the uploading demands of all devices. The normalized system cost with respect to the increasing of computation capability of the edge server (denoted as $F$) is given in Fig.~\ref{cost}(b), where $M=30$. In contrast to Fig.~\ref{cost}(a), we observe that the system cost of three schemes decreases with $F$. As $F$ increases , the system cost of ``Fully-uploading" significantly decreases, and outperforms ``Randomly-uploading" when $F \geq 2$ GHz since the edge-CNN training time is decreased. From Figs.~\ref{cost}(a)-(b), we observe that when the compression ratio is given, the proposed scheme always outperforms the two benchmarks by tracing the best computing-communication latency tradeoff.

\subsubsection{Energy-Latency-and-Learning Tradeoffs}
{The effect of the compression ratio ($\lambda$) on the inference error rate and devices' normalized weighted-sum cost is shown in Fig.~\ref{tradeoff}(a).} As $\lambda$ increases, the inference error is increased since some features are lost in the compressed latent vectors, which is harmful for the CNN training. On another note, with a larger $\lambda$, the completion latency and energy raised can be decreased since the burdens of wireless channel can be relieved. This is reflected as the energy-latency-and-learning tradeoffs, which indicate that the compression ratio is a critical parameter and should be chosen carefully  in the proposed framework.

{Furthermore, the choice of $\lambda$ to balance the tradeoff of the proposed framework is investigated in Fig.~\ref{tradeoff}(b), where the efficiency is defined as $\eta=\frac{\rm Inference \ accuracy}{\rm Weighted \ sum \ cost}$. It is observed that $\eta$ first increases significantly as the increase of $\lambda$ and then reaches the top when $\lambda=32$. This is because compared to the uncompressed case ($\lambda=1$), the size of the compressed model is reduced greatly without significant loss of inference accuracy when $\lambda=32$. As $\lambda$ keeps rising, the efficiency starts to decrease, which is due to the fact that the gain brought by the model compression is not significant enough. Therefore, to trace an optimal tradeoff, the compression ratio is prone to be set as $32$ if the inference accuracy meets the requirement.}

\begin{figure}[t] 
\centering
\captionsetup{font={footnotesize }}
\subfigure[]{
\includegraphics[width=2.86in,height=1.99in]{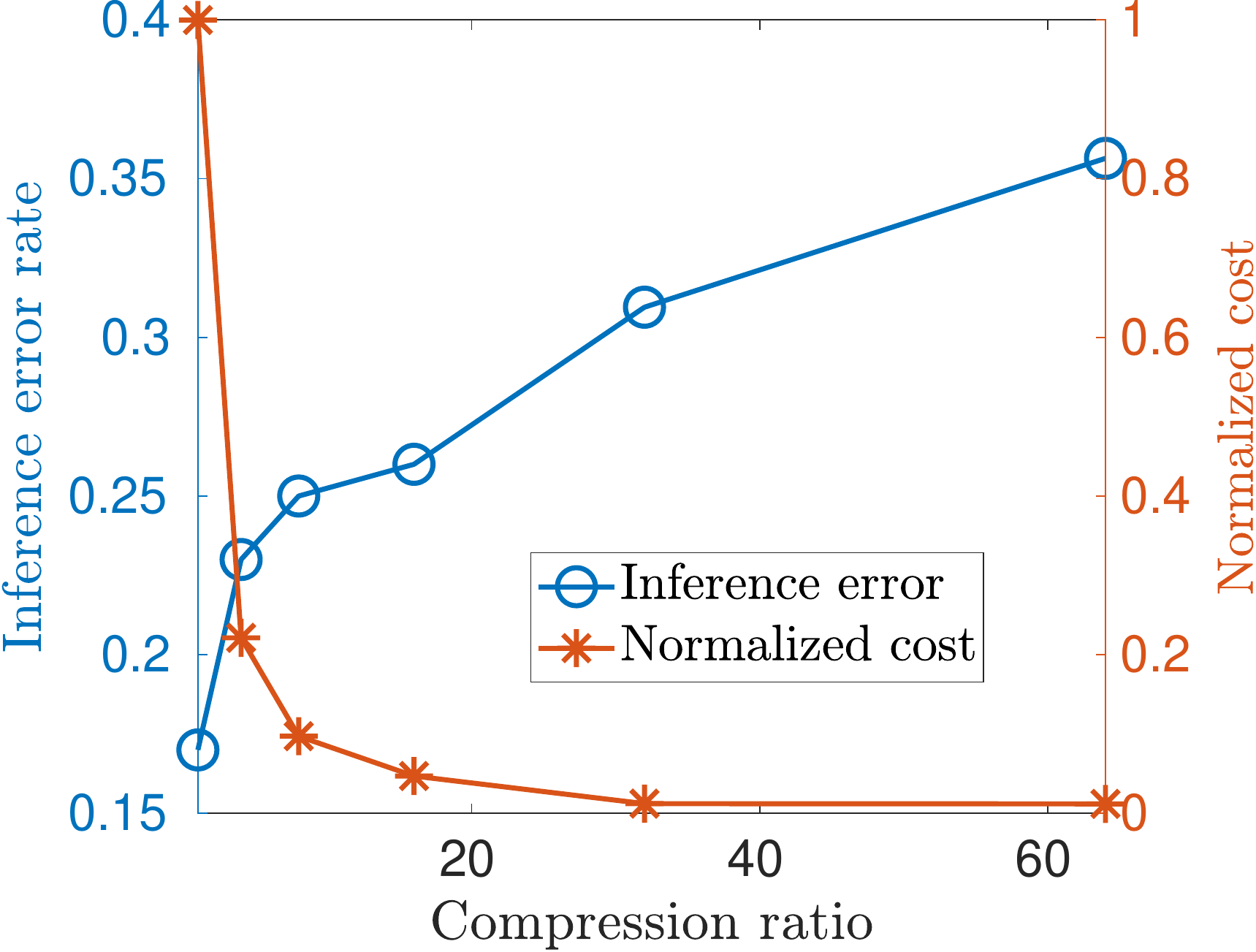}} 
\hspace{0.05in}
\subfigure[]{
\includegraphics[width=2.485in,height=1.99in]{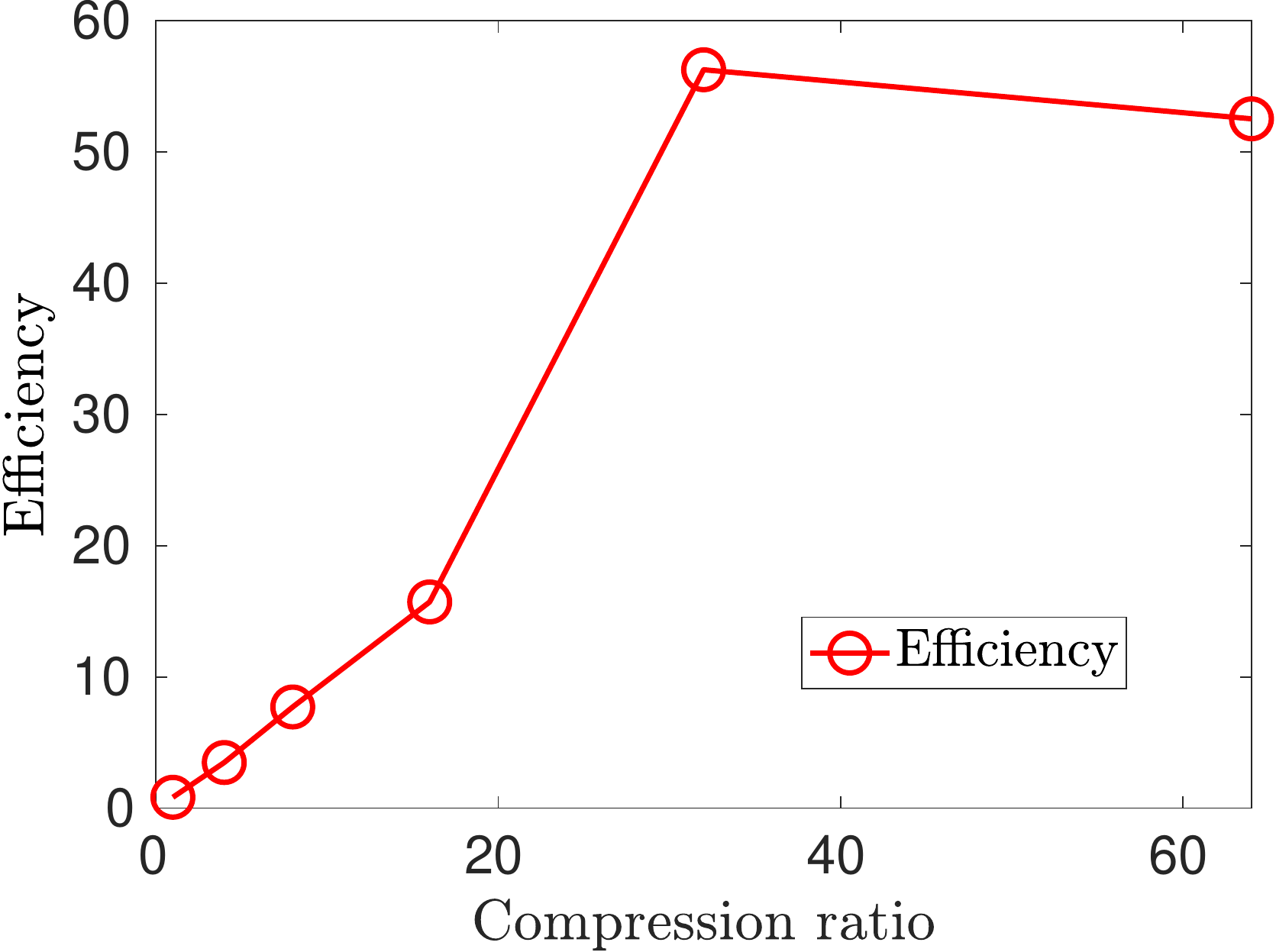}}
\caption{{Affect of compression ratio on the inference error rate and the normalized weighted-sum cost is shown in (a), and the effect of compression ratio on the efficiency is shown in (b),  where $\alpha=0.5$.}}
\label{tradeoff}
\end{figure}

\section{Concluding Remarks}
\label{concusion}
In this article, we developed a TL-empowered edge-CNN framework enabling the image classification with privacy-preserving over 5G IIoT edge networks. By revealing the energy-latency-and-learning tradeoffs, we formulated an optimization problem that jointly considers device uploading decision and wireless bandwidth allocation to minimize the weighted-sum cost. We derive the solution by decomposing the original problem into an uploading decision subproblem with a fixed bandwidth allocation and a wireless bandwidth allocation subproblem that minimizes the completion latency. 

Although the proposed framework has similar characteristics such as taking advantage of data locality from devices as in federated learning, however, this work outperforms the federated learning in three folds: 1) CNN training is performed at the edge and is not subject to the constraints of devices. 2) Autoencoder is locally trained independently without edge server involvement. 3) Privacy-preserving is obtained by transmitting latent vectors without additional cost of encryption.

%%%%%%%%%%%%%%%%%%%%%%%%%%%%%%%%%%%%%%%%%%%%%%%%%%%%%%%%%%%%%%%%%%%%%%%%%%%%%%%%%%%%%%%%%%%%%%

\end{document}